\documentclass[conference,letterpaper]{IEEEtran}
\usepackage[utf8]{inputenc} 
\usepackage[numbers,compress]{natbib}
\usepackage[cmex10]{amsmath}
\usepackage{amsfonts,mathrsfs,mathdots,mathtools,amsthm,amssymb,centernot}
\usepackage[bb=ams]{mathalfa}
\usepackage{algorithmic}
\usepackage{graphicx}
\usepackage{textcomp}
\usepackage[dvipsnames]{xcolor}
\usepackage{tikz}
\usepackage{pgfplots}
\usepackage{subcaption}
\usepackage{dsfont}
\usepackage{pifont}
\usepackage{siunitx}
\usepackage{comment}
\usepackage{pgfplotstable}
\usepackage[ruled]{algorithm2e}
\usepackage{color,soul}
\usepackage{bm}	
\usepackage[inline]{enumitem}
\setlist[enumerate,1]{label = \arabic*.,ref = \arabic*}
\usepackage{nicefrac}
\usepackage{subcaption}
\usepackage{colortbl}
\usepackage{booktabs}
\usepackage{diagbox}
\usepackage{url}
\usepackage{ifthen}
\usepackage{balance}
\usepackage{xfrac}
\usepackage{dirtytalk}
\usepackage{bbm}

\usepackage{hyperref}  
\hypersetup{colorlinks,
            linkcolor=blue,
            citecolor=blue,
            urlcolor=blue,
            anchorcolor=blue,
            filecolor=blue,
            linktocpage,
            plainpages=false,
            breaklinks=true}

\def\supp{\textnormal{supp}}

\theoremstyle{definition}
\newtheorem{definition}{Definition}
\newtheorem{lemma}{Lemma}
\newtheorem{theorem}{Theorem}

\newtheorem{corollary}{Corollary}
\newtheorem{proposition}{Proposition}
\newtheorem{example}{Example}
\newtheorem{remark}{Remark}

\mathtoolsset{showonlyrefs}
\allowdisplaybreaks
\IEEEoverridecommandlockouts
\begin{document}
\title{Dobrushin Coefficients of Private Mechanisms Beyond Local Differential Privacy} 

 \author{%
   \IEEEauthorblockN{Leonhard Grosse\IEEEauthorrefmark{1},
                     Sara Saeidian\IEEEauthorrefmark{1}\IEEEauthorrefmark{2},
                     Tobias J. Oechtering\IEEEauthorrefmark{1},
                      and Mikael Skoglund\IEEEauthorrefmark{1}%
   \IEEEauthorblockA{\IEEEauthorrefmark{1}%
                    KTH Royal Institute of Technology, Stockholm, Sweden,
                     \{lgrosse, saeidian, oech, skoglund\}@kth.se}
       \IEEEauthorblockA{\IEEEauthorrefmark{2}%
                    Inria Saclay, Palaiseau, France}}

\thanks{This work was supported by the Swedish Research Council (VR) under grants 2023-04787 and 2024-06615.} 
}

\maketitle

\begin{abstract}
 We investigate Dobrushin coefficients of discrete Markov kernels that have bounded pointwise maximal leakage (PML) with respect to all distributions with a minimum probability mass bounded away from zero by a constant $c>0$. This definition recovers local differential privacy (LDP) for $c\to 0$. We derive achievable bounds on contraction in terms of a kernel's PML guarantees, and provide mechanism constructions that achieve the presented bounds. Further, we extend the results to general $f$-divergences by an application of Binette's inequality. Our analysis yields tighter bounds for mechanisms satisfying LDP and extends beyond the LDP regime to \emph{any} discrete kernel.  
\end{abstract}


\section{Introduction}
In the local model of privacy, individual data points, modeled as realizations of a random variable $X$, are privatized by randomization through a Markov kernel (usually called the \emph{privacy mechanism}) $K: (\mathcal X,\mathcal Y) \to [0,1]$ which outputs the privatized random variable $Y$. Most commonly, the local privacy guarantee of a kernel is quantified by \emph{local differential privacy} (LDP) \cite{duchi2013LDPminmaxDEF,DPoriginalpaper}, a functional of the kernel defined by 
\begin{equation}
\label{eq:LDPdef}
    \text{LDP}(K) = \sup_{y\in\mathcal Y} \sup_{x,x' \in\mathcal X} \log\frac{K_{Y|X=x}(y)}{K_{Y|X=x'}(y)}.
\end{equation}
Whenever $\text{LDP}(K)\leq \varepsilon$, the mechanism is said to satisfy \emph{$\varepsilon$-LDP}.
While methods designed for LDP are increasingly used in practice (see, e.g., \cite{googleRAPPOR,appleDP, abowd2018us}), it has been observed that its definition may be overly restrictive for some statistical applications, leading to significant degradation in estimation performance. For example, \citet{duchi2013local} argue that in parameter estimation problems, the increase in risk introduced by LDP guarantees can be extreme when the privacy parameter is picked in a way that provides stringent guarantees. Subsequent related works \cite{asoodeh2024contraction,9517999,10206578} identify mechanisms' \emph{contraction coefficients} as important tools for determining such risk increase and related quantities introduced by privatization. Heuristically, the contraction coefficient $\eta_f(K)$ of a Markov kernel $K$ with respect to some $f$-divergence $D_f$ quantifies \say{how much closer} two distributions will become in $D_f$ after being passed through the kernel. More formally, the contraction coefficient of a kernel is defined as 
\begin{equation}
    \eta_f(K) = \sup_{P_X\neq Q_X} \frac{D_f(K\circ P_X||K\circ Q_X)}{D_f(P_X||Q_X)}.
\end{equation}
The coefficients can then be applied to find private minimax risk bounds by reduction to a hypothesis testing problem. For example, Le\,Cam's two point method \cite[Thm.~31.1]{Polyanskiy_Wu_2025} allows us to bound the minimax risk $R^\star$ in estimating a parameter $\theta \in \Theta$ of the random variable $X$ as
\begin{equation}
    \label{eq:LeCam}
     R^{\star} \coloneqq \inf_{\hat \theta}\sup_{\theta\in\Theta}\mathbb E[l(\hat \theta,\theta)]\geq \frac{l(\theta_0,\theta_1)}{2}\Big[1-\text{TV}(P_{\theta_0}||P_{\theta_1})\Big],
\end{equation}
for any two $\theta_0,\theta_1 \in \Theta$ and any loss function $l(\cdot,\cdot)$ satisfying a triangle inequality, where $P_\theta$ denotes the distribution of $X$ given $\theta$. If $X$ is further privatized by some Markov kernel, contraction coefficients along with convenient tensorization properties of $f$-divergences can be used to quantify the degradation in estimation performance due to privatization.
The study of contraction coefficients has long been of interest in the information-theory community, and many works on fundamental properties exists. For a detailed review of results, see e.g., \cite[Chapter 33]{Polyanskiy_Wu_2025}. One of the most fundamental results in this domain states that \emph{any} $f$-divergence contraction coefficient is upper bounded by the contraction coefficient of the total variation distance, an $f$-divergence with $f(t)=0.5|t-1|$. In the study of Markov chains, this quantity is commonly referred to as the \emph{Dobrushin coefficient} $\eta_\text{TV}$ \cite{dobrushin1956central}. 

In this work, we study Dobrushin coefficients of locally private mechanisms whenever LDP as a privacy measure turns out to be too strict. We are motivated to do so by a combination of observations: Firstly, as mentioned above, there exists situations in which effective privacy guarantees with LDP (that is, guarantees with a small enough privacy parameter) can only be made by accepting large increases in estimation risk. Secondly, the definition of LDP in \eqref{eq:LDPdef} demonstrates a somewhat converse issue: Whenever a (discrete) Markov kernel contains zero elements (whenever $K_{Y|X=x}(y)=0$ for some $y$), the LDP measure is infinite. However, this behavior is somewhat counter-intuitive in the broader perspective. To illustrate, consider the two Markov kernels,
\begin{equation}
   K_1= \begin{bmatrix}
        \frac{1}{3} & \frac{1}{3} & \frac{1}{3} & 0 \\
        0 & \frac{1}{3} & \frac{1}{3} & \frac{1}{3} \\
        \frac{1}{3} & 0 & \frac{1}{3} & \frac{1}{3} \\
        \frac{1}{3} & \frac{1}{3} & 0 & \frac{1}{3}
    \end{bmatrix},
    \quad K_2 = \begin{bmatrix}
        1 & 0 & 0 & 0 \\
        0 & 1 & 0 & 0 \\
        0 & 0 & 1 & 0 \\
        0 & 0 & 0 & 1
    \end{bmatrix}.
\end{equation}
Clearly, $K_2$ offers no privacy compared to simply releasing the realization of the secret $X$ directly, while $K_1$ significantly randomizes any input realization of $X$. Still, both $K_1$ and $K_2$ are equivalently \say{non-private} in the framework of LDP. 

In order to quantify local privacy in cases such as the above example, we instead employ \emph{pointwise maximal leakage} (PML) \cite{saeidian2023pointwise}, a recently proposed pointwise adaption of maximal leakage \cite{IssaMaxL,alvim2012measuring} in the framework of quantitative information flow \cite{alvim2020science}.
PML is built on concrete threat-models, and is therefore able to offer \emph{meaningful} privacy guarantees in terms of adversarial inference risks \cite{saeidian2023inferential}. Further, PML can be seen as a generalization of LDP. In particular, the privacy quantification of PML depends both on the mechanism $K$, as well as on the distribution $P_X$ of the private random variable $X$. The generalization follows since, as shown in \cite{IssaMaxL},
\begin{equation}
    \text{LDP}(K) = \sup_{P_X}\,\text{PML}(K,P_X),
\end{equation}
that is, guaranteeing a certain level of LDP is equivalent to guaranteeing PML for \emph{any} data-generating distribution. 

\subsection{Contributions}
In this paper, we present achievable bounds on the Dobrushin coefficients of mechanisms satisfying PML for any data-generating distribution $P_X$ with minimum probability mass bounded away from zero by some constant $c>0$, that is, for which $\min_x P_X(x)\geq c$. Further, we use the result to obtain strong data processing inequalities for general $f$-divergences. All omitted proofs can be found in the appendix.

\section{Preliminaries}
We use uppercase letters to denote random variables, lowercase letters to denote their realizations and caligraphic letters to denote sets. Specifically, we consider the (discrete) random variables $X$ and $Y$, where $X$ represents input data into a Markov kernel $K$ (a stochastic matrix also referred to as the \emph{(privacy) mechanism} below) that induces a random variable $Y$ at its output. Let $P_{XY}$ denote the joint distribution of $X$ and $Y$. Then we use $P_{XY} = K \times P_X$ to imply that $P_{XY}(x,y) = K_{Y|X=x}(y)P_X(x)$ and $P_{Y} = K \circ P_X$ to denote the marginalization $P_Y(y) = \sum_{x \in \mathcal X} K_{Y|X=x}(y)P_X(x)$. We write $\mathcal P(\mathcal X)$ for the probability simplex on $\mathcal X$. For $N\in\mathbb N$, we define $[N] \coloneqq \{1,\dots,N\}$ as the set of positive integers up to $N$. Finally, we use $\mathcal S_{N,M}$ to denote the set of $N\times M$ row-stochastic matrices, and $\log$ denotes the natural logarithm.

\subsection{$f$-Divergences and Contraction Coefficients}
For any convex function $f: (0,\infty) \to \mathbb R$ such that $f(1)=0$, we define the \emph{$f$-divergence} between two probability measures $P \ll Q$ \cite{csiszar1967information} as
\begin{equation}
    D_f(P||Q) = \mathbb E_Q\bigg[f\bigg(\frac{dP}{dQ}\bigg)\bigg].
\end{equation}
We write $D_f(\cdot ||\cdot)=\text{TV}(\cdot||\cdot)$ whenever $f(t) = 0.5|1-t|$, that is, if the $f$-Divergence of interest is the total variation distance. Other notable examples for specific choices of $f$ include relative entropy ($f(t) = t\log(t)$), in which case we write $D(\cdot||\cdot)$, and the squared Hellinger divergence ($f(t) = (1-\sqrt{t})^2$), which we denote by $H^2(\cdot||\cdot)$. For any $f$-divergence according to the above definition, the \emph{data-processing inequality} (DPI) states that $D_f(K\circ P_X||K\circ Q_X)\leq D_f(P_X||Q_X)$ for any Markov kernel $K$ and any two distributions $P_X$, $Q_X$. In order to further quantify the gap between the two sides of the DPI, consider the \emph{contraction coefficient} of a kernel $K$,
\begin{equation}
    \eta_{f}^{\mathcal P}(K) = \sup_{P_X,Q_X \in \mathcal P} \frac{D_f(K\circ P_X||K\circ Q_X)}{D_f(P_X||Q_X)},
\end{equation}
where $\mathcal P$ is some arbitrary subset of the probability simplex $\mathcal P\subseteq \mathcal P(\mathcal X)$. Whenever $\mathcal P = \mathcal P(\mathcal X)$, we drop the additional superscript and simply write $\eta_{f}(K)$. Naturally, the DPI already implies $\eta_f^{\mathcal P}(K) \leq 1$ always. We are hence mainly concerned with identifying cases in which $\eta_f^{\mathcal P}(K)<1$. In the majority of what follows, we will consider the contraction coefficient of a kernel with respect to the total variation distance, which due to \citet{dobrushin1956central} is expressed in computable form as,
\begin{equation}
\label{eq:Drobushin}
    \eta_\text{TV}(K) = \sup_{x\neq x'}\, \text{TV}(K_{Y|X=x}||K_{Y|X=x'}).
\end{equation}
It is known \cite{Polyanskiy_Wu_2025} that for any kernel $K$, $\eta_{\chi^2}(K)\leq \eta_f(K) \leq \eta_\text{TV}(K)$, where $\eta_{\chi^2}(K)$ denotes the case $f(t)=(t-1)^2$.

\subsection{Pointwise Maximal Leakage}
Pointwise maximal leakage (PML) \cite{saeidian2023pointwise} is a notion of privacy related to maximal leakage \cite{IssaMaxL,alvim2012measuring}, addressing a key limitation of the on-average maximal leakage measures: rare but revealing outputs of a mechanism can lead to severe underestimation of inference risk. PML quantifies leakage of individual outcomes $Y=y$, and can therefore capture worst-case posterior risks that otherwise may be hidden by averaging.

PML admits equivalent operational interpretations based on adversarial threat models. In particular, it can be characterized via a general gain-function framework in which an adversary seeks to maximize the expected value of non-negative gain functions after observing an outcome of the mechanism \cite{alvim2012measuring}.
\begin{definition}[Gain function view of PML {{\cite[Cor. 1]{saeidian2023pointwise}}}]
\label{def:PMLgainfunc}
    Let $X$ be a random variable defined on the set $\mathcal X$ distributed according to $P_X$. Let $(X,Y)$ be induced by $P_X$ togther with the mechanism $P_{Y|X}$. Then the \emph{pointwise maximal leakage in the gain function view from $X$ to an outcome $y$} is defined as 
    \begin{equation}
        \ell_{P_{XY}}(X\to y) \coloneqq \log \, \sup_g \frac{\sup_{P_{W|Y}}\mathbb E[g(X,W)\mid Y=y]}{\max_{w\in\mathcal W}\mathbb E[g(X,w)]},
    \end{equation}
    where $\mathcal W$ is some arbitrary set of \say{guesses}, and $g: \mathcal X \times \mathcal W \to \mathbb R_+$ is an arbitrary non-negative gain function.
\end{definition}
Definition \ref{def:PMLgainfunc} can be understood in the following way: Assume an adversary has access to the output $y$ of a mechanism $P_{Y|X}$. Using this outcome, she forms a guess $\hat w$ of an arbitrary property $f(X=x)=w$ of the realization $X=x$. She quantifies the quality of that guess by some gain function $g(x,w)$. The privacy cost of releasing $y$ through $P_{Y|X}$ is quantified by comparing the gain from an optimal guess made \emph{with} access to $y$, to the gain of the best blind guess \emph{without} access to the mechanism's output $y$. The PML of the mechanism is then obtained by taking the maximum privacy cost for \emph{any} (non-negative) gain function $g$, and any guessing space $\mathcal W$.
The PML framework includes a broad class of attacks, including reconstruction and membership inference attacks (see \cite{saeidian2023pointwise} for corresponding gain functions). By quantifying the inference risk for the \emph{worst-case} gain function, PML is a robust and meaningful measure grounded in adversarial success rates.\footnote{As shown in \cite{saeidian2023pointwise}, the same quantity also arises from a pointwise adaption of the randomized function model in \cite{IssaMaxL}.} For a mechanism $K$ and data distribution $P_X$ (with full support), the PML to an outcome $y$ is shown in \cite{saeidian2023pointwise} to simplify to,
\begin{equation}
    \ell_{K\times P_X}(X\!\to\! y)=\log\frac{\max_{x\in\mathcal X}K_{Y|X=x}(y)}{(K\circ P_X)(y)}.
\end{equation}
The following result relates LDP and PML.
\begin{proposition}[{{\cite[Theorem 14]{IssaMaxL}}}]
\label{prop:LDPisPMLc0}
    The LDP guarantee of any Markov kernel $K$ mapping from $\mathcal X$ to $\mathcal Y$ can be expressed as 
    \begin{equation}
    \label{eq:LDPisPMLforallPriors}
        \text{LDP}(K) = \sup_{y\in\mathcal Y}\sup_{P_X\in\mathcal P(\mathcal X)} \ell_{K\times P_X}(X\to y).
    \end{equation}
\end{proposition}

\subsection{Related Work}
Contraction coefficients of $f$-divergences are investigated in, e.g., \cite{ahlswede1976spreading,polyanskiy2015dissipation,ordentlich2021strong,raginsky2016strong,sason2019data,nishiyama2020relations,jin2024properties}. In the privacy domain, \citet{duchi2013local} show that if $K$ satisfies $\varepsilon$-LDP, then
\begin{equation}
    D(K\circ P_X||K\circ Q_X) \leq \min\{4,e^{2\varepsilon}\}(e^\varepsilon-1)^2 \text{TV}^2(P_X||Q_X).
\end{equation}
The authors further use this relation to reason about increase in minimax estimation risk introduced by a LDP constraint. As a byproduct of their construction of optimal mechanisms, \citet{kairouz2016extremal} show that if kernel $K$ is $\varepsilon$-LDP, then,
\begin{equation}
\label{eq:kairouzbound}
    \eta_\text{TV}(K) \leq \frac{e^\varepsilon-1}{e^\varepsilon+1}.
\end{equation}
Related, \citet{asoodeh2024contraction} show that,
\begin{equation}
\label{eq:AsoodehsSDPI}
    \eta_{\chi^2}(K) = \eta_\text{KL}(K) =\eta_{H^2}(K)\leq \bigg[\frac{e^\varepsilon-1}{e^\varepsilon+1}\bigg]^2.
\end{equation}
Furhter, \citet{10206578} bound the contraction coefficient of the elementary $E_\gamma$-divergence in terms of a kernels LDP guarantee.\footnote{Any $f$-divergence can be expressed as a function of this divergence \cite{cohen1998comparisons}.}
The authors further provide a non-linear strong data processing inequality for the same divergence. Bounds on contraction coefficients are shown to be useful for deriving results on private mixing times \cite{zamanlooy2024mathrm}, minimax estimation risk \cite{asoodeh2024contraction}, privacy amplification \cite{asoodeh2020privacy,grosse2025bounds} and the privacy guarantees of noisy gradient descent methods \cite{asoodeh2024privacy}.

\section{Dobrushin Coefficients with Bounded PML}
\label{sec:etaTV}
We define the following generalization of LDP based on the observation in Proposition \ref{prop:LDPisPMLc0}. Let for $0 < c \leq \nicefrac{1}{|\mathcal X|}$,
\begin{equation}
    \mathcal Q_\mathcal X(c) \coloneqq \{P_X \in\mathcal P(\mathcal X): \min_{x\in\mathcal X}P_X(x)\geq c\}.
\end{equation}
With this, we define $(\varepsilon,c)$-PML by restricting the supremum in \eqref{eq:LDPisPMLforallPriors} from $\mathcal P(\mathcal X)$ to the set $\mathcal Q_\mathcal X(c)$.
\begin{definition}
\label{def:epscPML}
    Let $X$ be a random variable defined on the space $\mathcal X$ with $N\coloneqq |\mathcal X|$. For some set $\mathcal P \subseteq \mathcal P(\mathcal X)$, define the \emph{$\mathcal P$-local leakage capacity} \cite{grosse2025privacy} of a kernel $K$ as
    \begin{equation}
        C(K,\mathcal P) \coloneqq \sup_{P_X \in\mathcal P} \sup_{y\in\mathcal Y}\, \ell_{K\times P_X}(X\to y).
    \end{equation}
    For some $\varepsilon\geq 0$ and $0 < c \leq \nicefrac{1}{N}$, we say that a kernel $K\in\mathcal S_{N,M}$ \emph{satisfies $(\varepsilon,c)$-PML} if, 
    \begin{equation}
        C(K,\mathcal Q_\mathcal X(c)) \leq \varepsilon \stackrel{\text{notation}}{\iff} K\in\mathcal M(\varepsilon,c).
    \end{equation}
\end{definition}
The above definition allows us to smoothly interpolate between an $\varepsilon$-PML for only the uniform distribution ($c=\nicefrac{1}{N})$), and $\varepsilon$-LDP $(c=0)$. Operationally, $(\varepsilon,c)$-PML limits the potential gain of any adversary according to Def.~\ref{def:PMLgainfunc} for any data-generating distribution in the set $\mathcal Q_\mathcal X(c)$. We remark that \emph{any} mechanism (even the identity mapping) satisfies $(-\log(c),c)$-PML. This implies that $(\varepsilon,c)$-PML can quantify privacy leakage for any mechanism, in particular also those not satisfying any finite LDP guarantee, like discrete mechanisms with zero probability assignments.
The following statement relates a kernel's zero-assignments $K_{Y|X=x}(y)=0$ to its $(\varepsilon,c)$-PML guarantee. This statement can be seen as a generalization of \cite[Lemma 2]{10646583}. Further, its proof stresses some of the convenient properties of PML guarantees in terms of disclosure limits, as they are discussed in detail in \cite{saeidian2023inferential}.

\begin{lemma}
\label{lem:privacyregion}
Assume $K\in\mathcal M(\varepsilon,c)$. If $\varepsilon < -\log\big((N-l)c\big)$ for some $l\in [N-1]$, then (the stochastic matrix representing) $K$ can have at most $l-1$ zero elements in each column.   
\end{lemma}

As the main object of interest in what follows, we define the maximal Dobrushin coefficient given an $(\varepsilon,c)$-PML guarantee for any $\varepsilon \geq 0$, $c \in (0,\nicefrac{1}{N}]$ as,
\begin{equation}
\label{eq:etaTVepscdef}
    \eta_\text{TV}(\varepsilon,c) \coloneqq \sup_{K\in\mathcal M(\varepsilon,c)}\eta_\text{TV}^{\mathcal Q_\mathcal X(c)}(K).
\end{equation}
 The following lemma shows that the restriction of input distributions has no influence on the value of $\eta_\text{TV}(\varepsilon,c)$. 

\begin{lemma}
    \label{lem:QcdontchangeTVcontr}
        For any $0< c <\nicefrac{1}{N}$, we have
        \begin{equation}
            \eta_\text{TV}(\varepsilon,c) = \sup_{K \in\mathcal M(\varepsilon,c)} \eta_\text{TV}(K) . 
        \end{equation}
\end{lemma}
Lemma \ref{lem:QcdontchangeTVcontr} implies that, when considering the contraction of the total variation distance with an $(\varepsilon,c)$-PML constraint, the restriction of input distributions only changes the set of feasible mechanisms, while the contraction coefficient itself takes the input-distribution-independent form given in \eqref{eq:Drobushin}. 

The main statement of this work is the following.
    \begin{theorem}
    \label{thm:summary}
        Let $0 < c \leq \nicefrac{1}{N}$ and $\varepsilon\geq0$. If a Markov kernel $K \in \mathcal M(\varepsilon,c)$, that is, if $K$ satisfies $(\varepsilon,c)$-PML, then, 
        \begin{equation}
            \eta_\text{TV}(\varepsilon,c) \leq \min\Bigg\{\frac{e^\varepsilon-1}{e^\varepsilon(1-Nc)+1},\,1\Bigg\}.
        \end{equation}
    \end{theorem}
The remainder of this section is dedicated to the proof of this theorem. We proceed in two distince steps.

First, we characterize the contraction coefficients of private kernels when the privacy-guarantee is not strict.
\begin{theorem}
\label{thm:fdivcontraction}
    Let $0<c\leq\nicefrac{1}{N}$. If $\varepsilon \geq \log(\nicefrac{2}{Nc})$, then
    \begin{equation}
        \sup_{K\in\mathcal M(\varepsilon,c)}\eta_f(K) = 1,
    \end{equation}
    for any convex and twice-differentiable $f$. 
\end{theorem}
\begin{IEEEproof}
    The key observation is that the mechanism $K = [p,1-p]$ with $p\in[0,1]^{N}$ defined by  $p = [\nicefrac{1}{2},\dots,\nicefrac{1}{2},1,0]^\top$
    satisfies $\log(\nicefrac{2}{Nc})$-PML. Such a mechanism induces a \emph{decomposable} pair of random variables, that is, $(X,Y)$ such that
    \begin{equation}
    \label{eq:conditionnocontraction}
        \exists x,x' \in \mathcal X: \supp\big(K_{Y|X=x}\big) \cap \supp\big(K_{Y|X=x'}\big) = \emptyset.
    \end{equation}
  Let  $\rho_m(X;Y)$ be the \emph{maximal correlation} between $X$ and $Y$. It is shown in, e.g., \cite[Thm. 33.6(c)]{Polyanskiy_Wu_2025} that,
    \begin{equation}
        \eta_{\chi^2}(K) = \sup_{P_X \in\mathcal P_{\mathcal X}} \rho_m^2(X;Y) \eqqcolon s(K). 
    \end{equation}
    It was shown by \citet{ahlswede1976spreading} that $s(K) < 1$ iff \eqref{eq:conditionnocontraction} is false. Together with the data processing inequality for $f$-divergences, this implies that whenever \eqref{eq:conditionnocontraction} holds, we have $\eta_{\chi^2}(K) = 1$. The claim follows since $\eta_{\chi^2}(K) = \eta_f(K)$ for any convex and twice-differentiable $f$ \cite[Thm 33.6(b)]{Polyanskiy_Wu_2025}.
\end{IEEEproof}
We remark that Theorem \ref{thm:fdivcontraction} bounds the contraction coefficient of $f$-divergences \emph{without} any restriction of the set of input distributions. The privacy constraint in Theorem \ref{thm:fdivcontraction} should therefore only be seen as a two-number summary of the channel $K$. In particular, if the set of input distributions is restricted, for example to $\mathcal Q_\mathcal X(c)$, some $f$-divergences might still contract. For total variation distance, Theorem \ref{thm:fdivcontraction} together with $\eta_{f}(K) \leq \eta_\text{TV}(K)\leq 1$ shows the latter part of the $\min$ of Theorem \ref{thm:summary}, that is, if $\varepsilon \geq \log (\nicefrac{2}{Nc})$, then $\eta_\text{TV}(\varepsilon,c)=1$. 

Next, we analyze the case $\varepsilon < \log(\nicefrac{2}{Nc})$. Note that the set $\mathcal M(\varepsilon,c)$ of mechanisms satisfying $(\varepsilon,c)$-PML makes no assumptions about the support size of the induced output variable $Y$, and therefore is a subset of $\bigcup_{L=2}^\infty \mathcal S_{N,L}$. In Lemma \ref{lem:cardinality} below, we show that it is enough to restrict the optimization in \eqref{eq:etaTVepscdef} to kernels that induce a binary output, that is, to $K\in \mathcal M(\varepsilon,c)\cap \mathcal S_{N,2}$. We also remark that \cite{ordentlich2021strong} shows that $\eta_f$ is achieved by input distributions with binary support. However, due to the restriction of input distributions to $\mathcal Q_\mathcal X(c)$, this result does not directly apply to the presented problem.
\begin{lemma}
\label{lem:cardinality}
    To find the value of $\eta_\text{TV}(\varepsilon,c)$, it is sufficient to consider mechanisms in $\mathcal M(\varepsilon,c)$ with $|\mathcal Y|=2$.
\end{lemma}
This lemma ensures that we can limit the optimizations to binary-output mechanisms in what follows. 
The following statement characterizes the optimal solution to \eqref{eq:etaTVepscdef} for the cases that fall outside of the scope of Theorem \ref{thm:fdivcontraction}. 
\begin{theorem}
\label{thm:M1bound}
    For any $N\geq 2$, let $c \in (0,\nicefrac{1}{N}]$ and assume $0\leq \varepsilon < \log\nicefrac{2}{Nc}$. Then,
    \begin{equation}
    \label{eq:optimalM1bound}
        \eta_\text{TV}(\varepsilon,c) \leq \frac{e^\varepsilon-1}{e^\varepsilon(1-Nc)+1}.
    \end{equation}
\end{theorem}
\begin{IEEEproof}
Due to Lemma \ref{lem:cardinality}, we consider the case $\mathcal Y = \{0,1\}$. Fix any $x,x'\in\mathcal X$. Let $S=\sum_{x\in\mathcal X} K_{Y|X=x}(0)$. For any $K\in\mathcal M(\varepsilon,c)$, due to Lemma \ref{lem:setcharacterization} in Appendix \ref{app:polytope}, 
\begin{equation}
    K_{Y|X=x}(0) \leq e^\varepsilon (cS +(1-Nc)K_{Y|X=x'}(0)),
\end{equation}
\begin{equation}
    1-K_{Y|X=x'}(0) \leq e^\varepsilon (c(N-S)+(1-Nc)(1-K_{Y|X=x}(0)).
\end{equation}
Summing up these two inequalities and rearranging, we obtain,
\begin{equation}
    (K_{Y|X=x}(0) - K_{Y|X=x'}(0))(1+e^\varepsilon(1-Nc)) \leq e^\varepsilon-1.
\end{equation}
Since the choice of $x,x' \in\mathcal X$ was arbitrary, this holds for any pair. In particular, we have,
\begin{align}
    \eta_\text{TV}(K) &= \max_{x\neq x'} \max_{y\in\{0,1\}}|K_{Y|X=x}(y)-K_{Y|X=x'}(y)| \\ &\leq\frac{e^\varepsilon-1}{1+e^\varepsilon(1-Nc)},
\end{align}
since if $\mathcal Y= \{0,1\}$, we have $|K_{Y|X=x}(0)-K_{Y|X=x'}(0)| = |K_{Y|X=x}(1)-K_{Y|X=x'}(1)|$ for all $x,x' \in\mathcal X$.
\end{IEEEproof}

\begin{remark}
    If $c\to0$, the bound in Theorem \ref{thm:M1bound} recovers the bound \eqref{eq:kairouzbound} for local differential privacy, as we expect due to the fact that $(\varepsilon,0)$-PML $\iff$ $\varepsilon$-LDP. On the other hand, if $c=\nicefrac{1}{N}$, we have $\eta_\text{TV}(\varepsilon,\nicefrac{1}{N})=e^\varepsilon-1$, which is also achieved by the optimal mechanisms designed with respect to the uniform distribution in \cite{10646583}. 
\end{remark}
\begin{remark}
\label{rem:optimalM}
    A binary-output construction achieving the bound in Theorem \ref{thm:M1bound} is the following: Let $q\in\{1,...,N-1\}$ and let 
    \begin{equation}
        m = \frac{1-e^\varepsilon c q}{1+e^\varepsilon(1-Nc)}, \quad \text{and } M =\frac{e^\varepsilon(1-cq)}{1+e^\varepsilon(1-Nc)}.
    \end{equation}
    From these values, the first column $p\in[0,1]^N$ of an optimal mechanism $K^* = [p,1-p]$ is given as,
    \begin{equation}
        p=[\underbrace{M,\dots,M}_{q\text{ times}},\underbrace{m\dots,m}_{N-q\text{ times}}]^\top.
    \end{equation}
\end{remark}
Theorem~\ref{thm:fdivcontraction} and Theorem~\ref{thm:M1bound} together show the statement in Theorem~\ref{thm:summary}. We conclude this section with an example.
\begin{example}[Binary mechansim]
\label{ex:binary}
    If $N=2$, the only choice in Remark \ref{rem:optimalM} is $q=1$. Hence, an optimal mechanism for $(\varepsilon,c)$-PML in the binary-input case is
    \begin{equation}
        K = \frac{1}{e^\varepsilon(1-2c)+1}\begin{bmatrix}
            e^\varepsilon(1-c) & 1-e^\varepsilon c \\
            1-e^\varepsilon c & e^\varepsilon(1-c)
        \end{bmatrix}.
    \end{equation}
    This coincides with the optimal binary mechanism for PML in an $\ell_1$-uncertainty set of radius $1-2c$, as presented in \cite[Theorem 5]{grosse2025privacy}.\footnote{Note that for $|\mathcal X|=2$, $\mathcal Q_\mathcal X(c)$ is an $\ell_1$-ball with radius $1-2c$.} For $c\to 0$, $K$ simplifies to Warner's randomized response \cite{warnerRRoriginal}, the optimal binary mechanism for $\varepsilon$-LDP \cite{kairouz2016extremal}.
\end{example}

\section{Extension to General $f$-Divergences}
\label{sec:fdiv}
Finally, we show how the above bounds on the private Dobrushin coefficient can be used to obtain bounds on general $f$-divergences. The application is based on the $f$-divergence inequality originally presented by \citet{binette2019note}, and related approaches are taken in \cite{hirche2024quantum,nuradha2025non,grosse2025bounds}. In particular, we apply Binettes generalization of reverse Pinsker's inequality \cite[Theorem 1]{binette2019note} \say{across} a kernel satisfying $(\varepsilon,c)$-PML, giving,
\begin{align}
\label{eq:binette}
    &D_f(K\circ P_X||K\circ Q_X) \\&\leq\eta_\text{TV}(\varepsilon,c)\Bigg[\frac{f\big(\Gamma_{\min}(\varepsilon,c)\big)}{1-\Gamma_{\min}(\varepsilon,c)}+\frac{f\big(\Gamma_{\max}(\varepsilon,c)\big)}{\Gamma_{\max}(\varepsilon,c)-1}\Bigg]\delta,
\end{align}
with $\delta \coloneqq \text{TV}(P_X||Q_X)$ and $\Gamma_{\max}$, $\Gamma_{\min}$ defined by
\begin{equation}
    \Gamma_{\max}(\varepsilon,c) \coloneqq \sup_{P_X,Q_X \in \mathcal Q_\mathcal X(c)} \sup_{K\in\mathcal M(\varepsilon,c)}\sup_{y\in\mathcal Y}\frac{(K\circ P_X)(y)}{(K\circ Q_X)(y)},
\end{equation}
\begin{equation}
    \Gamma_{\min}(\varepsilon,c) \coloneqq \inf_{P_X,Q_X \in \mathcal Q_\mathcal X(c)} \inf_{K\in\mathcal M(\varepsilon,c)}\inf_{y\in\mathcal Y}\frac{(K\circ P_X)(y)}{(K\circ Q_X)(y)}.
\end{equation}
The following proposition bounds these extremal quantities in terms of $\varepsilon$ and $c$, and therefore enables us to obtain a private bound on any $f$-divergence in the $(\varepsilon,c)$-PML case via \eqref{eq:binette}.
\begin{proposition}
\label{prop:Gammabound}
   Let $c\in(0,\nicefrac{1}{N}]$, $0\leq \varepsilon\leq -\log(c)$. Whenever $K\in \mathcal M(\varepsilon,c)$, we have $\Gamma_{\max}(\varepsilon,c) \leq (1-Nc)e^\varepsilon + 1$ and $\Gamma_{\min}(\varepsilon,c) \geq \big((1-Nc)e^\varepsilon+1\big)^{-1}$.
\end{proposition}
The following examples illustrate the bounds obtained by this technique. Let $\Xi(\varepsilon,c) \coloneqq \min\big\{\frac{e^\varepsilon-1}{e^\varepsilon(1-Nc)+1},\,1\big\}$ below.
\begin{corollary}[Relative Entropy]
\label{corr:relativeentropy}
    For any $c\in(0,\nicefrac{1}{N}]$, $\varepsilon\geq 0$, assume the kernel $K$ satisfies $(\varepsilon,c)$-PML. Then we have for any $P_X,Q_X \in \mathcal Q_\mathcal X(c)$ with $\text{TV}(P_X||Q_X)\leq \delta$,
    \begin{align}
       &D(K\circ P_X||K\circ Q_X) \leq \Xi(\varepsilon,c)\log\Big((1-Nc)e^\varepsilon+1\Big)\delta.
    \end{align}
\end{corollary}

\begin{corollary}[Hellinger divergence]
\label{corr:hellinger}
    For any $c\in(0,\nicefrac{1}{N}]$, $\varepsilon\geq 0$, assume the kernel $K$ satisfies $(\varepsilon,c)$-PML. Then we have for any $P_X,Q_X \in \mathcal Q_\mathcal X(c)$ with $\text{TV}(P_X||Q_X)\leq \delta$,
    \begin{align}
        H^2&(K\circ P_X|| K\circ Q_X) \\&\leq \Xi(\varepsilon,c)\bigg(2-\frac{4}{\sqrt{(1-Nc)e^\varepsilon+1}+1}\bigg)\delta.
    \end{align}
\end{corollary}

The bounds presented in the above corollaries are illustrated in Figure \ref{fig:corrbounds} for two different kernels $K_1$ and $K_2$, which we define in the following example.
\begin{example}[Divergence bounds]
\label{ex:RR}
Let $N=10$ and consider the mechanism $K_1 \in \mathcal M(\log\nicefrac{10}{3},0.05)$ defined by
\begin{equation}
    K_1 = \begin{bmatrix}
        \nicefrac{15}{16} & \dots & \nicefrac{15}{16} &\nicefrac{1}{16} & \dots & \nicefrac{1}{16} \\
        \nicefrac{1}{16} & \dots & \nicefrac{1}{16} &\nicefrac{15}{16} & \dots & \nicefrac{15}{16}
    \end{bmatrix}^\top,
\end{equation}
further, for $N=5$, consider $K_2 \in \mathcal M(\log\nicefrac{10}{3},0.1)$ given by,
\begin{equation}
    K_2 = \begin{bmatrix}
        \nicefrac{1}{3} & \nicefrac{1}{3} & \nicefrac{1}{3} & 0 & 0 \\
        0 & \nicefrac{1}{3} & \nicefrac{1}{3} & \nicefrac{1}{3} & 0 \\
        0 & 0 & \nicefrac{1}{3} & \nicefrac{1}{3} & \nicefrac{1}{3} \\
        \nicefrac{1}{3} & 0 & 0 & \nicefrac{1}{3} & \nicefrac{1}{3} \\
        \nicefrac{1}{3} & \nicefrac{1}{3} & 0 & 0 & \nicefrac{1}{3}
    \end{bmatrix}.
\end{equation}
Note that $K_1$ also satisfies $\log(15)$-LDP. In Figure \ref{fig:corrbounds}(a), we compare our presented upper bound in Corollary \ref{corr:relativeentropy} to the bound presented by \citet{duchi2013local}, the bound presented by \citet[Theorem 5]{10206578} and the bound obtained by the contraction coefficient of relative entropy presented by \citet{asoodeh2024contraction} with an additional application of reverse Pinsker's inequality in \cite[Theorem 28]{7552457} (with $Q_{\min}=c$) to obtain a bound between $D(K\circ P_X||K\circ Q_X)$ and $\text{TV}(P_X||Q_X)$. Since $K_2$ does not satisfy any finite LDP guarantee, such a comparison is not possible in this case. The bound in Corollary \ref{corr:hellinger} is shown in Figure \ref{fig:corrbounds}(b) for $K_2$.
\end{example}

Finally, the following example shows how the obtained bounds can be applied in a minimax estimation setting.

\begin{example}[Private minimax risk]
\label{ex:bern_iid}
Let $X_1,\dots,X_n$ be i.i.d. with unknown parameter
$\theta\in\{\theta_0,\theta_1\}$, where $\theta_1-\theta_0=\delta\in(0,\nicefrac{1}{4})$ and $P_{\theta_0} = [\nicefrac{1}{4}-\delta,\nicefrac{1}{4}+\delta,\nicefrac{1}{4},\nicefrac{1}{4}]^\top$ and $P_{\theta_1}=[\nicefrac{1}{4}+\delta,\nicefrac{1}{4}-\delta,\nicefrac{1}{4},\nicefrac{1}{4}]^\top$.
Each sample is released through a mechanism
$K\in\mathcal M(\varepsilon,c)$, producing privatized i.i.d.~observations
$Y_i\sim K\circ X_i$. Applying Corollary~\ref{corr:relativeentropy},
    \begin{align}
    D(K\circ P_{\theta_0}\|K\circ P_{\theta_1})
    \leq
    \Xi(\varepsilon,c)\log\!\big((1-4c)e^\varepsilon+1\big)\,\delta .
\end{align}
By the tensorization property of relative entropy \cite{Polyanskiy_Wu_2025},
\begin{align}
D\big((K\circ P_{\theta_0})^{\otimes n}||&
(K\circ P_{\theta_1})^{\otimes n}\big)
\\&\leq
n\,\Xi(\varepsilon,c)\log\!\big((1-4c)e^\varepsilon+1\big)\,\delta .
\end{align}
Le Cam's two-point method together with the Bretagnolle-Huber inequality \cite{bretagnolle2006estimation} implies the minimax risk bound,
\begin{equation}
R^\star
\geq
\frac12\exp\!\Big(
- n\,\Xi(\varepsilon,c)\log\!\big((1-4c)e^\varepsilon+1\big)\,\delta
\Big).
\end{equation}
In particular, to achieve $R^\star\le 1/4$, we require up to $\log$-terms,
\begin{equation}
n
=
\Omega\!\left(
\frac{1}{\Xi(\varepsilon,c)\,\delta}
\right) .
\end{equation}
This shows that $(\varepsilon,c)$-PML mechanisms reduce the effective sample size
by a factor proportional to the private Dobrushin coefficient $\Xi(\varepsilon,c)$. If $\varepsilon\geq \log\nicefrac{1}{(2c)}$, we have $\Xi(\varepsilon,c)=1$, and the sample complexity is equal to the non-private setup.
\end{example}


\begin{figure}
    \centering
    \begin{subfigure}[b]{\linewidth}
    \centering
    \label{fig:corrbounds:subfig:corr1}
        \includegraphics[scale=0.75]{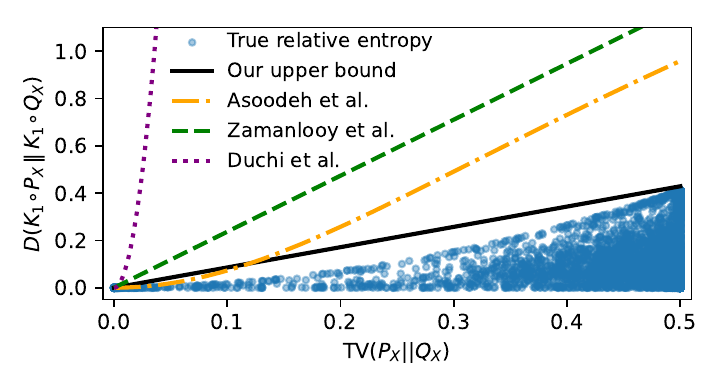}
        \caption{$K_1$, $N=10$, $c=0.05$.}
    \end{subfigure}
    \begin{subfigure}[b]{\linewidth}
    \centering
    \label{fig:corrbounds:subfig:corr2}
        \includegraphics[scale=0.75]{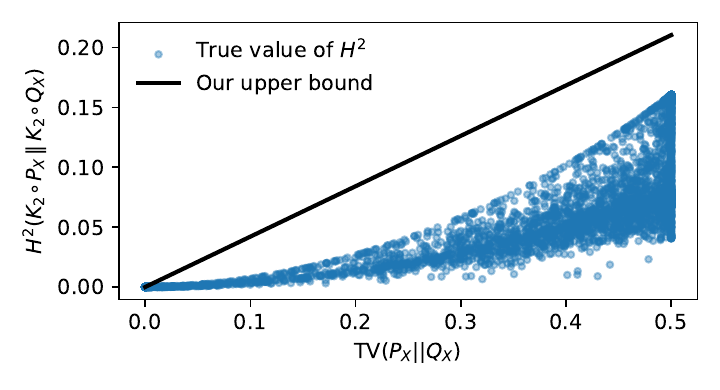}
        \caption{$K_2$, $N=5$, $c=0.1$.}
    \end{subfigure}
    \caption{Numerical evaluation of the bounds presented in Corollaries \ref{corr:relativeentropy} and \ref{corr:hellinger} for randomly generated distributions in the set $\mathcal Q_\mathcal X(c)$. $K_i$ for $i=1,2$ are given in Example \ref{ex:RR}. }
    \label{fig:corrbounds}
\end{figure}

\bibliographystyle{IEEEtranN}
\footnotesize
\balance
\bibliography{main}

\normalsize
\appendices
\onecolumn
\section{Proof of Lemma \ref{lem:privacyregion}}
The proof is based on the following observation about the disclosure of subset-memberships.
\label{app:privregionproof}
\begin{proposition}
\label{prop:disclosureprevgeneralization}
    Assume $X$ is distributed according to $P_X$. Let $\mathcal X_{\text{sub}} \subseteq \mathcal X$ and $P_X(\mathcal X_{\text{sub}}) = \sum_{x\in\mathcal X_{\text{sub}}}P_X(x)$. If $K$ satisfies $\varepsilon$-PML with $\varepsilon < -\log P_X(\mathcal X_{\text{sub}})$, then $K$ cannot deterministically reveal the fact that a realization $x \in \mathcal X_{\text{sub}}$. 
\end{proposition}
\begin{proof}
    We show the statement by contra-positive. Assume that $K$ deterministically reveals $x\in\mathcal X_{\text{sub}}$ for some outcomes $X=x \in \supp(P_X)$ and $Y=y \in \supp(P_Y)$. Fix this specific $y$. We have $K_{X|Y=y}(x') = 0$ for all $x' \notin \mathcal X_{\text{sub}}$. Hence, we also have
    \begin{equation}
        K_{Y|X=x'}(y) = \frac{K_{X|Y=y}(x')P_Y(y)}{P_X(x')} = 0 \quad \forall x'\notin \mathcal X_{\text{sub}}.
    \end{equation}
    Now, define the distribution $Q_X$ by
    \begin{equation}
        Q_X(x) = \begin{cases}
            \frac{P_X(x)}{P_X(\mathcal X_{\text{sub}})}, &\text{if } x \in \mathcal X_{\text{sub}} \\
            0, &\text{otherwise.}
        \end{cases}
    \end{equation}
    Then we have 
    \begin{align}
        \ell_{K\times P_X}(X \to y) &= \log \frac{\max_{x\in\mathcal X}K_{Y|X=x}(y)}{\sum_{x\in\mathcal X}K_{Y|X=x}(y)P_X(x)}\\[.5em]
        &= \log\frac{\max_{x\in\mathcal X_{\text{sub}}}K_{Y|X=x}(y)}{\sum_{x\in\mathcal X_{\text{sub}}}K_{Y|X=x}(y)P_X(x)}\\[.5em]
        &= \log \frac{\max_{x\in\mathcal X_{\text{sub}}}K_{Y|X=x}(y)}{P_X(\mathcal X_{\text{sub}})\sum_{x\in\mathcal X_{\text{sub}}}K_{Y|X=x}(y)Q_X(x)}\label{subeeq:xsubappears}\\[.5em]
        &= \log\frac{1}{P_X(\mathcal X_{\text{sub}})} + \mathcal \ell_{K\times Q_X}(X\to y) \\[.5em]
        &\geq -\log P_X(\mathcal X_{\text{sub}}) \label{subeq:elldissapears},
    \end{align}
    Where \eqref{subeeq:xsubappears} follows from the definition of $Q_X$ above, and \eqref{subeq:elldissapears} follows from the nonnegativity of PML \cite[Lemma 1]{saeidian2023pointwise}.
\end{proof}
  
Now, assume that $K$ has $l$ zero elements in one of its columns, say in the column corresponding to output symbol $y$. Let $\mathcal X_{\text{sub}} = \{x\in \mathcal X: K_{Y|X=x}(y) > 0\}$. Then the mechanism deterministically discloses subset-membership to $\mathcal X_{\text{sub}}$ via the outcome $y$. (Whenever an adversary observes the outcome $Y=y$, she can deterministically infer that the input realization $X=x$ was in $\mathcal X_{\text{sub}}$.) According to Proposition \ref{prop:disclosureprevgeneralization} in Appendix \ref{app:privregionproof}, this implies that $\ell_{K\times P_X}(X\to y) \geq -\log P_X(\mathcal X_\text{sub})$ for any distribution $P_X$. We further have,
    \begin{equation}
        \min_{P_X \in\mathcal Q_\mathcal X(x)} P_X(\mathcal X_\text{sub}) = (N-l)c.
    \end{equation}
    This implies that $C(K,\mathcal Q_\mathcal X(c)) \geq -\log\big((N-l)c\big)$. \qed 

\section{Proof of Lemma \ref{lem:QcdontchangeTVcontr}}
\label{app:TVeqonSimplexlemma}
        From the extreme-point representation of the set $\mathcal Q_\mathcal X(c)$,
        \begin{equation}
            \mathcal Q_\mathcal X(c) =\{c + V_X \mid V_X \in \mathcal P(\mathcal X)\},
        \end{equation}
        where $c + P_X$ is understood as adding scalar $c$ to every component of $V_X$. We have that 
        \begin{align}
            \eta^{\mathcal Q_\mathcal X(x)}_{\text{TV}}(K) &= \sup_{P_X,Q_X \in\mathcal Q_\mathcal X(c)} \frac{||K\circ P_X - K \circ Q_X ||_1}{||P_X-Q_X||_1} \\[.7em]
            &= \sup_{V_X,W_X \in \mathcal P(\mathcal X)}\frac{||K\circ (c + (1-Nc)V_X) - K \circ (c + (1-Nc)W_X) ||_1}{||c + (1-Nc)V_X-c - (1-Nc)W_X||_1}.
        \end{align}
        Now, note that, 
        \begin{equation}
            ||c + (1-Nc)V_X-c - (1-Nc)W_X||_1 = (1-Nc)||V_X-W_X||_1, 
        \end{equation}
        and,
        \begin{equation}
            ||K\circ (c + (1-Nc)V_X) - K \circ (c + (1-Nc)W_X) ||_1 = (1-Nc)||K \circ V_X - K \circ W_X||_1.
        \end{equation}
        Therefore, we may write,
        \begin{align}
            \eta^{ \mathcal Q_\mathcal X(c)}_{\text{TV}}(K) &= \sup_{V_X,W_X \in \mathcal P(\mathcal X)} \frac{(1-Nc)||K\circ V_X-K \circ W_X||_1}{(1-Nc)||V_X - W_X||_\text{TV}} \\
            &= \sup_{V_X,W_X \in \mathcal P(\mathcal X)} \frac{||K\circ V_X-K \circ W_X||_1}{||V_X - W_X||_\text{TV}} = \eta_{\text{TV}}(K).
        \end{align}
        Hence, the restriction from $\mathcal P(\mathcal X)$ to $\mathcal Q_\mathcal X(c)$ does not affect $\eta_\text{TV}$. \qed

\section{Proof of Lemma \ref{lem:cardinality}}
\label{app:proofcardinality}
   Assume $K \in \mathcal M(\varepsilon,c)$ induces output alphabet $\mathcal Y$ with cardinality $|\mathcal Y|>2$. Fix two arbitrary distributions $P_X,Q_X \in \mathcal Q_\mathcal X(c)$ and pick $E^* \subseteq \mathcal Y$ such that $\text{TV}(P_Y||Q_Y) = P_Y(E^*)-Q_Y(E^*)$. Define the random variable $Z$ on the binary alphabet $\mathcal Z\coloneqq \{0,1\}$ as being induced by the kernel $K'$ defined as
   \begin{equation}
       K'_{Z|Y=y}(z) = \begin{cases}
           1, &\text{if }z=0 \text{ and }y\in E^* \text{ or }z=1\text{ and }y \in (E^*)^c, \\
           0, &\text{ otherwise.}
       \end{cases}
   \end{equation}
   With this, the Markov chain $X-Y-Z$ holds (for both $P_X$ and $Q_X$). Let $K^{\text{BIN}}$ be the kernel mapping from $\mathcal X$ to $\mathcal Z$ by the successive application of $K$ and $K'$, that is, let $K^{\text{BIN}}=K'\circ K$. Clearly, $K^\text{BIN}$ has binary output support and by the post-processing property of PML \cite[Lemma 1]{saeidian2023pointwise}, we have,
   \begin{equation}
       \max_{z\in\{0,1\}}\ell(X\to z) \leq \max_{y\in\mathcal Y}\ell (X \to y),
   \end{equation}
   for any input distribution of $X$. Therefore $K^\text{BIN} \in \mathcal M(\varepsilon,c)$. Additionally, let $P_Z = K^\text{BIN} \circ P_X$, $Q_Z = K^\text{BIN}\circ Q_X$. We have
   \begin{align}
       \text{TV}(P_Z||Q_Z) &= \max_{z\in\{0,1\}}(P_Z(z)-Q_Z(z)) \\[.5em]
       &=\max\Big\{P_Y(E^*)-Q_Y(E^*),P_Y\big((E^*)^c\big)-Q_Y\big((E^*)^c\big)\Big\} \\[.5em]
       &=P_Y(E^*) - Q_Y(E^*) \\[.5em]
       &=\text{TV}(P_Y||Q_Y).
   \end{align}
   Since this applies to arbitrary input distributions $P_X$, $Q_X$, it also applies to the maximizing distributions in 
   \begin{equation}
       \eta_\text{TV}(K) = \sup_{P_X\neq Q_X} \frac{\text{TV}(K\circ P_X||K\circ Q_X)}{\text{TV}(P_X||Q_X)},
   \end{equation}
   showing that with the corresponding choice of $K'$, $\eta_\text{TV}(K) = \eta_\text{TV}(K^\text{BIN})$. Since $K \in \mathcal M(\varepsilon,c)$ implies that $K^\text{BIN}\in \mathcal M(\varepsilon,c)$, this proves the claim.
    \qed

\section{$(\varepsilon,c)$-PML as Linear Inequality Constraints}
\label{app:polytope}
\begin{lemma}
\label{lem:setcharacterization}
    For any $\varepsilon\geq 0$ and any $c \in (0,\nicefrac{1}{N}]$, $K \in \mathcal M(\varepsilon,c)$ implies that for all $y\in\mathcal Y$,
    \begin{align}
        K_{Y|X=x}(y) \leq e^\varepsilon\bigg(c\sum_{x'\in\mathcal X}K_{Y|X=x'}(y) + (1-Nc)K_{Y|X=x''}(y)\bigg), \quad \forall \, x,x'' \in \mathcal X.
    \end{align}
\end{lemma}
\begin{proof}
     Assume that $K \in \mathcal M(\varepsilon,\mathcal Q_\mathcal X(c))$. Then we have that $C(K,\mathcal Q_\mathcal X(c)) \leq \varepsilon$, that is, 
    \begin{align}
        &\quad \quad \quad \sup_{P_X\in\mathcal Q_\mathcal X(c)} \max_{y\in\mathcal Y} \frac{\max_{x\in\mathcal X}K_{Y|X=x}(y)}{\sum_{x'\in\mathcal X}P_X(x')K_{Y|X=x'}(y)} &\leq e^\varepsilon \\[1em]
        &\iff \sup_{W_X \in\mathcal P(\mathcal X)} \max_{y\in\mathcal Y} \frac{\max_{x\in\mathcal X}K_{Y|X=x}(y)}{\sum_{x'}(c+(1-Nc)W_X(x))K_{Y|X=x'}(y)} &\leq e^\varepsilon \\[1em]
        &\iff \max_y \frac{\max_x K_{Y|X=x}(y)}{c\sum_{x'} K_{Y|X=x'}(y) + (1-Nc)\inf_{W_X \in\mathcal P(\mathcal X)}\sum_{x''}W_X(x'')K_{Y|X=x''}(y)} &\leq e^\varepsilon \\[1em]
        &\iff \max_{y} \frac{\max_x K_{Y|X=x}(y)}{c\sum_{x'} K_{Y|X=x'}(y) + (1-Nc)\min_{x''}K_{Y|X=x''}(y)} &\leq e^\varepsilon \\[1em]
        &\iff \, \forall y\in\mathcal Y:\,\, \frac{\max_{x\in\mathcal X}K_{Y|X=x}}{c\sum_{x'}K_{Y|X=x}(y) + (1-Nc)\min_{x''}K_{Y|X=x''}(y)} &\leq e^\varepsilon 
    \end{align}
    This implies the statement of the lemma by rearranging and relaxing the minimum and maximum on $\mathcal X$. 
\end{proof}
\section{Proof of Proposition \ref{prop:Gammabound}}
\label{app:proofGammabound}
We have that 
\begin{equation}
    \Gamma_{\max} \leq\max_{y\in\mathcal Y} \frac{\sup_{P_X\in\mathcal Q_{\mathcal X}(x)}\sum_{x\in\mathcal X}K_{Y|X=x}(y)P_X(x)}{\inf_{Q_X\in\mathcal Q_{\mathcal X}(x)}\sum_{x\in\mathcal X}K_{Y|X=x}(y)Q_X(x)}.
\end{equation}
From the structure of $\mathcal Q_\mathcal X(c)$ we know that,
\begin{equation}
    \sup_{P_X\in\mathcal Q_{\mathcal X}(x)}\sum_{x\in\mathcal X}K_{Y|X=x}(y)P_X(x) = c\sum_{x\in\mathcal X}K_{Y|X=x}(y) + (1-Nc)\max_{x\in\mathcal X}K_{Y|X=x}(y)
\end{equation}
and 
\begin{equation}
    \inf_{Q_X\in\mathcal Q_{\mathcal X}(x)}\sum_{x\in\mathcal X}K_{Y|X=x}(y)Q_X(x) = c\sum_{x\in\mathcal X}K_{Y|X=x}(y) + (1-Nc)\min_{x\in\mathcal X}K_{Y|X=x}(y).
\end{equation}
Further, Lemma \ref{lem:setcharacterization} provides a bound on $\max_x K_{Y|X=x}(y)$, and hence we can bound,
\begin{equation}
    \sup_{P_X\in\mathcal Q_{\mathcal X}(x)}\sum_{x\in\mathcal X}K_{Y|X=x}(y)P_X(x) \leq c\sum_{x\in\mathcal X}K_{Y|X=x}(y) + (1-Nc) e^\varepsilon\bigg(c\sum_{x\in\mathcal X}K_{Y|X=x}(y)+(1-Nc)\min_{x\in\mathcal X}K_{Y|X=x}(y)\bigg).
\end{equation}
Therefore, we can bound $\Gamma_{\max}$ by 
\begin{equation}
    \Gamma_{\max} \leq (1-Nc)e^\varepsilon + \frac{c \sum_{x\in\mathcal X}K_{Y|X=x}(y)}{c \sum_{x\in X}K_{Y|X=x}(y)+(1-Nc)\min_{x\in\mathcal X}K_{Y|X=x}(y)}.
\end{equation}
The desired bound then follows by noticing that $\min_{x}K_{Y|X=x}(y) \geq 0$. The bound on $\Gamma_{\min}$ follows by the same steps.\qed

\end{document}